\documentclass[preprint]{elsarticle}

\usepackage{color,amssymb,amsmath}
\usepackage{amsthm}
\usepackage{tikz}
\usepackage{url}
\usepackage{subfigure}
\usetikzlibrary{calc}
\usepackage[hmargin=3cm,vmargin=3cm]{geometry}
\usepackage[color=green!30]{todonotes}
\tikzset{
  bigblue/.style={circle, draw=blue!80,fill=blue!40,thick, inner sep=1.5pt, minimum size=5mm},
  bigred/.style={circle, draw=red!80,fill=red!40,thick, inner sep=1.5pt, minimum size=5mm},
  bigblack/.style={circle, draw=black!100,fill=black!40,thick, inner sep=1.5pt, minimum size=5mm},
  bluevertex/.style={circle, draw=blue!100,fill=blue!100,thick, inner sep=0pt, minimum size=2mm},
  redvertex/.style={circle, draw=red!100,fill=red!100,thick, inner sep=0pt, minimum size=2mm},
  blackvertex/.style={circle, draw=black!100,fill=black!100,thick, inner sep=0pt, minimum size=2mm},  
  whitevertex/.style={circle, draw=black!100,fill=white!100,thick, inner sep=0pt, minimum size=2mm},  
  smallblack/.style={circle, draw=black!100,fill=black!100,thick, inner sep=0pt, minimum size=1mm},  
}

\def\acc#1{\left\{ #1 \right\}}

\renewcommand{\le}{\leqslant}
\renewcommand{\ge}{\geqslant}

\newcommand{\ecto}{\to}
\newcommand{\sto}{\ensuremath{\stackrel{s}{\longrightarrow}}}

\newcommand{\shom}[1]{\textsc{s-Hom}\ensuremath{(#1)}}
\newcommand{\HOM}[1]{\textsc{Hom}\ensuremath{(#1)}}
\newcommand{\Pshom}[1]{\textsc{Planar s-Hom}\ensuremath{(#1)}}
\newcommand{\PHOM}[1]{\textsc{Planar Hom}\ensuremath{(#1)}}

\newtheorem{theorem}{Theorem}[section]
\newtheorem{prop}[theorem]{Proposition}
\newtheorem{cor}[theorem]{Corollary}

\newtheorem{defn}[theorem]{Definition}
\newtheorem{lemma}[theorem]{Lemma}
\newtheorem*{thm:dicho}{Theorem~\ref{thm:dicho}}

\begin{document}

\begin{frontmatter}

\title{Complexity of planar signed graph homomorphisms to cycles}

\author[I3S]{Fran\c{c}ois Dross\fnref{anr}}
\ead{francois.dross@googlemail.com}
\author[labri,lifo]{Florent Foucaud\fnref{anr,ifcam}}
\ead{florent.foucaud@gmail.fr}
\author[irif]{Valia Mitsou\fnref{anr}}
\ead{vmitsou@liris.cnrs.fr}
\author[lirmm]{Pascal Ochem\corref{cor}\fnref{anr}}
\ead{pascal.ochem@lirmm.fr}
\author[labri]{Th\'eo Pierron\fnref{anr,ifcam}}
\ead{theo.pierron@labri.fr}

\address[I3S]{I3S, Universit\'e Nice Sophia Antipolis, France}
\address[labri]{Univ. Bordeaux, Bordeaux INP, CNRS, LaBRI, UMR5800}
\address[lifo]{Univ. Orl\'eans, INSA Centre Val de Loire, LIFO EA 4022, F-45067 Orl\'eans Cedex 2, France}
\address[irif]{IRIF, Universit\'e de Paris, France}
\address[lirmm]{LIRMM, Universit\'e de Montpellier, CNRS, Montpellier, France}

\fntext[anr]{This work is supported by the ANR project HOSIGRA (ANR-17-CE40-0022).}
\fntext[ifcam]{This author is partially supported by the IFCAM project ``Applications of graph homomorphisms'' (MA/IFCAM/18/39).}

\begin{abstract}
  We study homomorphism problems of signed graphs from a computational
  point of view. A signed graph is an undirected graph where each edge
  is given a sign, positive or negative. An important concept when
  studying signed graphs is the operation of switching at a vertex,
  which is to change the sign of each incident edge. A homomorphism of
  a graph is a vertex-mapping that preserves the adjacencies; in the
  case of signed graphs, we also preserve the edge-signs. Special
  homomorphisms of signed graphs, called s-homomorphisms, have been
  studied. In an s-homomorphism, we allow, before the mapping, to
  perform any number of switchings on the source signed graph. The
  concept of s-homomorphisms has been extensively studied, and a full
  complexity classification (polynomial or NP-complete) for
  s-homomorphism to a fixed target signed graph has recently been
  obtained. Nevertheless, such a dichotomy is not known when we
  restrict the input graph to be planar, not even for non-signed graph
  homomorphisms.

  We show that deciding whether a (non-signed) planar graph admits a
  homomorphism to the square $C_t^2$ of a cycle with $t\ge 6$, or to
  the circular clique $K_{4t/(2t-1)}$ with $t\ge2$, are
  NP-complete problems. We use these results to show that deciding whether a
  planar signed graph admits an s-homomorphism to an unbalanced even cycle is NP-complete. (A cycle is unbalanced if it has an odd number of negative edges). We deduce a
  complete complexity dichotomy for the planar s-homomorphism problem with any
  signed cycle as a target.

  We also study further restrictions involving the maximum
  degree and the girth of the input signed graph. We prove that planar s-homomorphism problems to signed cycles remain NP-complete even for inputs of maximum degree~$3$ (except for the case of unbalanced $4$-cycles, for which we show this for maximum degree~$4$). We also show that for a given integer $g$, the problem for signed bipartite planar inputs of girth $g$ is either trivial or NP-complete.
\end{abstract}


\end{frontmatter}

\section{Introduction}

In this paper, we study the computational complexity of graph and
signed graph homomorphism problems. Our main focus is the case where
the inputs are planar and the targets are cycles (but we also consider
other cases). Homomorphisms are structure-preserving mappings between
discrete structures; this type of problems is very general and models
many combinatorial problems. Consequently, the study of the
algorithmic properties of homomorphism problems is a rich area of
research that has gained a lot of attention. We refer to the
book~\cite{HNbook} as a reference on graph homomorphism problems.

An edge-coloured graph is a graph with several types of (undirected)
edges: each type corresponds to a colour. Given two edge-coloured
graphs $G$ and $H$, a \emph{homomorphism} of $G$ to $H$ is a
vertex-mapping $f$ of $V(G)$ to $V(H)$ that preserves adjacencies and
edge-colours, that is, if $x$ and $y$ are adjacent via a $c$-coloured
edge in $G$, then $f(x)$ and $f(y)$ must be adjacent via a
$c$-coloured edge in $H$ as well. When such a homomorphism exists, we
write $G\to H$. This concept is studied for example
in~\cite{AM98,Bthesis,HKRS01,MPPRS10,OPS17}. In this language,
standard undirected graphs can simply be seen as $1$-edge-coloured
graphs. Signed graphs are a special type of $2$-edge-coloured graphs
whose edge-colours are signs: positive and negative. In this paper, we
will consider two types of objects: standard undirected graphs (that will simply be
called \emph{($1$-edge-coloured) graphs}), and signed graphs.

\paragraph{Computational homomorphism problems} 

The most fundamental class of algorithmic homomorphism problems is the
following one (where $H$ is any fixed edge-coloured graph):

\medskip\noindent
\HOM{H} \\
Instance: An (edge-coloured) graph $G$.\\
Question: Does $G$ admit a homomorphism to $H$?
\medskip

Problem \HOM{H} has been studied for decades. For example, consider
$1$-edge-coloured graphs, and denote the $3$-vertex complete graph by
$K_3$. Then, \HOM{K_3} is the classic \textsc{$3$-Colouring} problem,
shown NP-complete by Karp~\cite{K72}. (More generally, a proper
$t$-colouring of a graph $G$ is a homomorphism to the complete graph
$K_t$.) \HOM{H} for edge-coloured graphs is studied for example
in~\cite{Bthesis,B94,BFHN17,BH00,MO17}.

\HOM{H} is also studied under the name of \textsc{Constraint
  Satisfaction Problem} (CSP), see for example~\cite{FV98}. In this
setting, edge-coloured graphs are seen as structures coming with a
number of symmetric binary relations (one for each edge-colour). In the context of CSPs, one also
considers discrete relational structures that may have relations of
arbitrary arities. The celebrated \emph{Dichotomy Conjecture} of Feder
and Vardi~\cite{FV98} and the subsequent work aims at classifying the
complexity of general CSP problems. While the conjecture was recently
solved in~\cite{B17,Z17} (independently) using the tools and language
of \emph{universal algebra}, this algebraic formulation does not
always provide simple explicit descriptions of the dichotomy. Thus,
obtaining explicit dichotomy classifications for relevant special
cases is still of major interest.

When studying \HOM{H}, we may always restrict ourselves to
edge-coloured graphs $H$ that are cores: $H$ is a \emph{core} if it
does not have any homomorphism to a proper subgraph of itself (in
other words, all of its endomorphisms are automorphisms). Moreover,
\emph{the} core of an edge-coloured graph $H$, noted $\text{core}(H)$,
is the smallest subgraph of $H$ that is a core. It is well-known that
the core of an edge-coloured graph is unique (up to isomorphism). It
is not difficult to observe that the complexity of \HOM{H} is the same
as the one of \HOM{\text{core}(H)}. For more details on these notions
see the book~\cite{HNbook}.

One of the most fundamental results in the area of CSP dichotomies is
the one obtained for $1$-edge-coloured graphs by Hell and
Ne\v{s}et\v{r}il. They proved in~\cite{HN90} that if the core of an
undirected graph $H$ has at least two edges, \HOM{H} is NP-complete,
and polynomial-time otherwise. (Note that the latter condition equates to say that $H$ is bipartite or has a loop, a property that is easily
testable.) One may ask the following natural question: what is the
behaviour of the above dichotomy for specific restrictions of the
input graphs?

\paragraph{Planar instances} 

Colouring and homomorphism problems are studied extensively for the
class of planar graphs. One consequence of the Four Colour Theorem
(that any planar graph is properly $4$-colourable) is that, contrarily
to the general case, \HOM{K_4} is trivially polynomial-time solvable
for planar instances. In this paper, we will study the following
restriction of \HOM{H}:

\medskip\noindent
\PHOM{H}\\
Instance: A planar (edge-coloured) graph $G$.\\
Question: Does $G$ admit a homomorphism to $H$?
\medskip

A complexity dichotomy for \PHOM{H} in the case of $1$-edge-coloured
graphs, if it exists, is probably not as easy to describe as the
Hell-Ne\v{s}et\v{r}il dichotomy for \HOM{H}. \PHOM{H} is known to be
NP-complete for $H=K_3$~\cite{GJ79}, but, as mentioned before, it is trivially
polynomial-time when $H$ contains a $4$-clique. There are other non-trivial examples
that are polynomial-time. For example, consider the Clebsch graph
$C_{16}$, a remarkable triangle-free graph of order~$16$. It follows
from~\cite{G12,N07} that every triangle-free planar graph has a
homomorphism to $C_{16}$. Since $C_{16}$ itself has no triangle, a
planar graph maps to $C_{16}$ if and only if it is triangle-free, and
thus \PHOM{C_{16}} is polynomial-time solvable. Infinitely many such
examples are known, see~\cite{HNT06,NO06}.

\PHOM{H} for $1$-edge-coloured graphs was studied more extensively
in~\cite{HNT06,MS09}, where it was independently proved to be
NP-complete when $H$ is any odd cycle $C_{2k+1}$, via two different
techniques. It is also proved in~\cite{MS09} that \PHOM{H} is
NP-complete whenever $H$ is subcubic and has girth~$5$. \PHOM{H} is
also proved NP-complete in~\cite{HNT06} when $H$ is an odd wheel, or
the Penny graph.

\paragraph{Further instance restrictions} 

The difficulty of classifying the complexity of \PHOM{H} makes it
meaningful to further refine the pool of graph instances to be
examined.  For example, restrictions on the maximum degree are studied
in~\cite{GHN00,S09}. In~\cite{S09}, it is proved that for every
undirected non-bipartite graph $H$, there is an integer $b(H)$ such
that \HOM{H} is NP-complete for graphs of maximum degree $b(H)$. The
value of $b(H)$ can be arbitrarily large, but for many graphs $H$,
$b(H)=3$~\cite{S09}. In particular, it is proved in~\cite{GHN00} that
for all $k\geq 1$, $b(C_{2k+1})=3$, that is, \HOM{C_{2k+1}} is
NP-complete for subcubic graphs. In fact, the result from~\cite{GHN00}
(combined with~\cite{HNT06,MS09}) also implies that \PHOM{C_{2k+1}} is
NP-complete for subcubic graphs.

Other restrictions are on the \emph{girth} of the input graph, that
is, the smallest length of a cycle. It is known that for any
$k\geq 1$, there is an integer $g=g(k)$ such that all planar graphs of
girth at least $g$ admit a homomorphism to $C_{2k+1}$~\cite{GGH01}. A
restriction to planar graphs of a conjecture of Jaeger~\cite{J88} implies that $g(k)=4k$, and it is
known that $4k\leq g(k)\leq (20k-2)/3$~\cite{BKKW04}.
  In~\cite{EMOP12}, it is proved that for every
fixed $k\geq 2$ and $g\geq 3$, either every planar graph of girth at
least $g$ maps to $C_{2k+1}$, or \PHOM{C_{2k+1}} is NP-complete for
such graphs. Other examples of this type of ``hypothetical
complexity'' theorems exist in other contexts, see for
example~\cite{FS,K98}.

\paragraph{Signed graphs and switching homomorphisms}

Signed graphs are special types of $2$-edge-coloured graphs, whose
edge-colours represent signs: positive and negative. Formally, a signed graph is a pair $(G,\sigma)$, where $G$ is the underlying graph (the
graph containing both the positive and the negative edges) and
$\sigma:E(G)\to\{-1,+1\}$ is the sign function that describes the
edge-signs. Signed graphs were studied as early as the 1930's in the
first book on graph theory by K\"onig~\cite{K36}. They were
rediscovered in the 1950's by Harary~\cite{H53}, who introduced the name
\emph{signed graph} and applied them to the area of social psychology. The concept was later developed by Zaslavsky in~\cite{Z82b} and numerous subsequent
papers~\cite{Z81,Z82b,Z82a,Z82c,Z84} and has become an important part
of combinatorics, with many connections to deep results and
conjectures. See~\cite{Zsurvey} for a dynamic bibliography on the
topic.

Zaslavsky~\cite{Z82b} introduced the \emph{switching} operation: given
a signed graph $G$ and a vertex $v$, to switch at $v$ is to change the
sign of all edges incident to $v$ (this can be seen as multiplying
their signs by $-1$). We say that two signed graphs $G$ and $G'$ are
\emph{switching-equivalent} if $G'$ can be obtained from $G$ by any
sequence of switchings. Note that one can test switching-equivalence
in polynomial time (on labelled graphs), see~\cite{BFHN17,Z82b}.

The switching operation turns out to be important in the context of
graph minors, and it relates to some outstanding problems in graph
theory, see~\cite{G05,NRS13,NRS14} for details. In this context,
Guenin introduced in~\cite{G05} a special kind of homomorphisms of
signed graphs, whose theory was later developed
in~\cite{NRS13,NRS14}. Following the terminology in~\cite{BFHN17}, we
define an \emph{s-homomorphism} of a signed graph $G$ to a signed
graph $H$ as a vertex-mapping $f$ from $V(G)$ to $V(H)$ such that
there exists a signed graph $G'$ that is switching-equivalent to $G$,
and $f$ is a classic edge-coloured graph homomorphism of $G'$ to
$H$. (Note that additionally one may allow switching at $H$; this does not change the problem, and we will generally not do so, but we can fix a convenient switching-equivalent sign function of $H$ and stick to it.) When an s-homomorphism exists, we note $G\sto H$.

Similarly to edge-coloured graph homomorphism, we say that a signed
graph $G$ is an \emph{s-core} if $G$ admits no s-homomorphism to a
proper signed subgraph of itself, and \emph{the} s-core of $G$ is the
smallest subgraph of $G$ that is an s-core (it is unique up to
s-isomorphism and switching~\cite{NRS14}).

We next define the decision problem that corresponds to
s-homomorphisms (with $H$ a fixed signed graph).

\medskip\noindent
\shom{H} \\
Instance: A signed graph $G$.\\
Question: Does $G$ admit an s-homomorphism to $H$?
\medskip

Note that for two switching-equivalent signed graphs $H$ and $H'$, the
definition of an s-homomorphism implies that \shom{H} and \shom{H'}
have the same complexity.

Extending the Hell-Ne\v{s}et\v{r}il dichotomy~\cite{HN90} for \HOM{H} for $1$-edge-coloured graph problems, a complexity dichotomy for
\shom{H} problems was proved in the papers~\cite{BFHN17,BS18}. The
authors showed that if the s-core of a signed graph $H$ has at least
three edges, then \shom{H} is NP-complete; it is polynomial-time
otherwise. On the other hand, it was shown in~\cite{BFHN17} that a
similar dichotomy for \HOM{H} problems for signed graphs (that is,
$2$-edge-coloured graphs and no switching allowed), is as difficult to obtain as the one for
general CSPs. This indicates that it probably cannot be stated in
simple graph-theoretic terms.

The authors of~\cite{BFHN17} asked what is the complexity of \shom{H}
problems when the input is planar. Let us define the corresponding
analogue of \PHOM{H}.

\medskip\noindent
\Pshom{H}\\
Instance: A planar signed graph $G$.\\
Question: Does $G$ admit an s-homomorphism to $H$?
\medskip

An interesting case is the one when $H=(C_k,\sigma)$ is a cycle
($k\geq 3$). If $H$ is switching-equivalent either to an all-positive
$C_{k}$, or to an all-negative $C_{k}$, then the complexity of
\shom{H} and \Pshom{H} are the same as the ones of \HOM{C_k} and
\PHOM{C_{k}}, respectively~\cite{BFHN17}. When $k$ is odd, one of
these two cases holds, and thus by the results of~\cite{HNT06,MS09},
in that case \Pshom{H} is NP-complete.

To state the situation when $k$ is even, it is convenient to introduce
the notion of \emph{balance}: a signed graph is \emph{balanced} if
every cycle contains an even number of negative edges. This central
notion is already present in the work of K\"onig~\cite{K36} but was
theorized by Harary~\cite{H53}.

If $k=2t$ is even and $H=(C_k,\sigma)$ is balanced, $H$ is switching-equivalent to
a positive $2$-vertex complete graph and thus
\Pshom{H} is polynomial-time. When $H$ is unbalanced, by switching $H$ if
necessary, we may assume that $H$ has a unique negative edge and
denote this signed graph by $UC_{2t}$. Then, \shom{UC_{2t}} has been
proved to be NP-complete~\cite{BFHN17,FN14}, but the complexity of
\Pshom{UC_{2t}} is not known. We settle this question in this paper.

We point out that homomorphisms of non-signed graphs to odd cycles are an important topic in the theory of homomorphisms and circular colourings. Odd cycles are among the simplest non-trivial graphs (with respect to homomorphisms) and are fundamental in the study of graph colourings. Several well-studied open questions and conjectures involving homomorphisms to odd cycles exist in the area, such as Jaeger's conjecture and others (see~\cite{BKKW04,EMOP12} for more details). For s-homomorphisms of bipartite signed graphs, unbalanced even cycles play a similar role as odd cycles for homomorphisms of non-signed graphs. Here also, certain interesting conjectures and theorems are stated, see~\cite{CNS} for a recent study. This motivates the study of the complexity of s-homomorphisms to cycles.

\paragraph{Our results} 

Our main goal is to prove that \Pshom{UC_{2k}} is NP-complete whenever
$k\geq 2$. As a first step, in Section~\ref{sec:non-signed}, we study non-signed graphs. We prove
that \PHOM{H} is NP-complete when $H$ is the square of a cycle; in
turn, this is used to prove that \PHOM{H} is NP-complete when $H$ is a
cubic circular clique. In Section~\ref{sec:uc2k}, we use these results to prove that
\Pshom{UC_{2k}} is NP-complete whenever $k\geq 3$. In
Section~\ref{sec:uc4}, using a different technique, we prove that the
case $k=2$ is also NP-complete. In Section~\ref{sec:girth}, we show
that for every integer $k\geq 1$ and even integer $g\geq 2$, either
every planar bipartite signed graph of girth $g$ admits a homomorphism
to $UC_{2k}$, or \Pshom{UC_{2k}} is
NP-complete for planar bipartite inputs of girth $g$. In Section~\ref{sec:maxdeg}, we
show that the results of Section~\ref{sec:uc2k} also apply to subcubic
input signed graphs (except for $UC_4$, for which this holds for maximum degree~$4$). We first start with some preliminary
considerations in Section~\ref{sec:prelim}.

\section{Preliminaries}\label{sec:prelim}

This section gathers some preliminary considerations.

\subsection{Some definitions}

Given a graph $G$, the \emph{square} of $G$, denoted $G^2$, is the
graph obtained from $G$ by adding edges between all vertices at
distance~$2$.

Given two integers $p$ and $q$ with $\gcd(p,q)=1$, the \emph{circular
  clique} $K_{p/q}$ is the graph on vertex set
$\{k_0,\ldots, k_{p-1}\}$ with $k_i$ adjacent to $k_j$ if and only if
$q\leq |i-j|\leq p-q$. Circular cliques are defined in the context of
circular chromatic number, see for example~\cite{Zhu99}.

\subsection{Switching graphs}

We now describe a construction that is important when studying
s-homomorphisms.

\begin{defn} Let $G$ be a signed graph. The \emph{switching graph} of
  $G$ is a signed graph denoted $\rho(G)$ and constructed as follows.
\begin{itemize}
\item[(i)] For each vertex $u$ in $V(G)$ we have two vertices $u_0$
  and $u_1$ in $\rho(G)$.
\item[(ii)] For each edge $e$ between $u$ and $v$ in $G$, we have four
  edges between $u_i$ and $v_j$ ($i,j \in \{ 0, 1 \}$) in $\rho(G)$,
  with the edges between $u_i$ and $v_i$ having the same sign as $e$
  and the edges between $u_i$ and $v_{1-i}$ having the opposite sign
  ($i \in \{0,1\}$).
\end{itemize}
\end{defn}

See Figure~\ref{fig1} for examples of signed graphs and their
switching graphs. (In all our figures, dashed edges are red/negative, while full edges are blue/positive). The notion of switching graph was defined by
Brewster and Graves in~\cite{BG09} in a more general setting related
to permutations (they called it \emph{permutation graph}). A related
construction was used in~\cite{KM04} in the context of digraphs. The
construction is also used in~\cite{OPS17} under the name
\emph{antitwinned graph}. Switching graphs play a key role in the
study of signed graph homomorphisms, indeed they have several useful
properties. One such property is that the switching graph of a signed
graph contains, as subgraphs, all switching-equivalent signed
graphs. Additionally, the following proposition allows us to study
s-homomorphisms in the realm of standard homomorphisms.

\begin{prop}[\cite{BFHN17}]
  \label{prop:permequiv}
  Let $G$ and $H$ be two signed graphs. Then, $G \sto H$ if and only if $G \ecto \rho(H)$.
\end{prop}

Thus, we obtain the following corollary.

\begin{cor}\label{cor:permequiv}
  Let $H$ be a signed graph. Then, \shom{H} and \Pshom{H} have the
  same complexity as \HOM{\rho(H)} and \PHOM{\rho(H)}, respectively.
\end{cor}

\subsection{The indicator construction}

We recall the \emph{indicator construction}, one of the main tools
used in the proof of the Hell-Ne\v{s}et\v{r}il dichotomy for \HOM{H}
in~\cite{HN90}.

\begin{defn}
  Let $H$ be a signed graph. An \emph{indicator} $(I,i,j)$, is a
  signed graph $I$ with two distinguished vertices $i$ and $j$ such
  that $I$ admits an automorphism mapping $i$ to $j$ and
  vice-versa. The \emph{result of the indicator $(I,i,j)$} applied to
  $H$ is an undirected graph denoted $H^*$ and defined as follows.
  \begin{itemize}
  \item[(i)] $V(H^*) = V(H)$
  \item[(ii)] There is an edge from $u$ to $v$ in $H^*$ if there is a
    homomorphism of $I \ecto H$ such that $i \mapsto u$ and
    $j\mapsto v$.
\end{itemize}

We say that $(I,i,j)$ \emph{preserves planarity} if, given a planar
undirected graph $G$, replacing each edge $uv$ with a copy of
$(I,i,j)$ by identifying $u$ with $i$ and $v$ with $j$, we obtain a
planar signed graph.
\end{defn}

The following result shows how we can use this tool.

\begin{theorem}[Hell and Ne\v{s}et\v{r}il~\cite{HN90}]
  \label{thm:indicator}
  Let $H$ be a signed graph, $(I,i,j)$, an indicator and $H^*$, the
  undirected graph resulting from $(I,i,j)$ applied to $H$. Then,
  \HOM{H^*} admits a polynomial-time reduction to
  \HOM{H}. 

  Moreover, if the indicator construction preserves planarity, then
  \PHOM{H^*} admits a polynomial-time reduction to \PHOM{H}.
\end{theorem}


\begin{proof}
  We sketch the proof. Given an input graph $G$ of \HOM{H^*}, we
  construct a signed graph $f(G)$ by replacing each edge $uv$ in $G$
  by a copy of $(I,i,j)$ by identifying $u$ with $i$ and $v$ with
  $j$. (If $G$ is planar and $(I,i,j)$ preserves planarity, then
  $f(G)$ is also planar.) Now it is not difficult to show that
  $G\to H^*$ if and only if $f(G)\to H$.
\end{proof}

As an example, consider the signed graph $UC_4$ and its switching
graph $\rho(UC_4)$. Let $I$ be the $4$-cycle with two parallel
negative and two parallel positive edges, where $i$ and $j$ are two
non-adjacent vertices $i$ and $j$. The result $\rho(UC_4)^*$ of
$(I,i,j)$ applied to $\rho(UC_4)$ is shown in Figure~\ref{fig1}: it
consists of two disjoint copies of $K_4$ (thus its core is $K_4$). By
Theorem~\ref{thm:indicator} and Corollary~\ref{cor:permequiv}, this
implies that \HOM{\rho(UC_4)} and \shom{UC_4} are NP-complete, by a
reduction from \HOM{K_4}. Note that the application of $(I,i,j)$
preserves planarity; however it is useless to reduce \PHOM{K_4} to
\PHOM{UC_4} since the former is polynomial-time solvable by the Four
Colour Theorem. We thus handle this case with an ad-hoc proof in
Section~\ref{sec:uc4}.

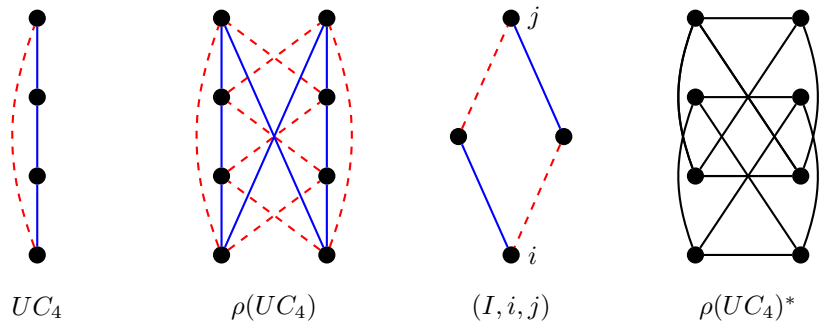
\begin{figure}[!htpb]
\begin{center}
\begin{tikzpicture}[every loop/.style={},scale=0.7]
  \node[blackvertex] (u) at (0,0) {};
  \node[blackvertex] (v) at (0,1.5) {};
  \node[blackvertex] (w) at (0,3) {};
  \node[blackvertex] (x) at (0,4.5) {};
  \draw[thick,blue] (u)--(v)--(w)--(x);
  \draw[thick,red,dashed] (u) to[bend left=20] (x);
  
  \node at (0,-1) {$UC_4$};

\begin{scope}[xshift=3.5cm]
  \node[blackvertex] (u0) at (0,0) {};
  \node[blackvertex] (v0) at (0,1.5) {};
  \node[blackvertex] (w0) at (0,3) {};
  \node[blackvertex] (x0) at (0,4.5) {};
  \draw[thick,blue] (u0)--(v0)--(w0)--(x0);
  \draw[thick,red,dashed] (u0) to[bend left=20] (x0);

  \node[blackvertex] (u1) at (2,0) {};
  \node[blackvertex] (v1) at (2,1.5) {};
  \node[blackvertex] (w1) at (2,3) {};
  \node[blackvertex] (x1) at (2,4.5) {};
  \draw[thick,blue] (u1)--(v1)--(w1)--(x1);
  \draw[thick,red,dashed] (u1) to[bend right=20] (x1);

  
  \draw[thick,red,dashed] (u0)--(v1)--(w0)--(x1) (u1)--(v0)--(w1)--(x0);
  \draw[thick,blue] (u0)--(x1) (u1)--(x0);
  
  \node at (1,-1) {$\rho(UC_4)$};
\end{scope}

\begin{scope}[xshift=9cm]
  \node[blackvertex] (i) at (0,0) {};
  \draw (i) node[right=1mm] {$i$};
  \node[blackvertex] (j) at (0,4.5) {};
  \draw (j) node[right=1mm] {$j$};
  \node[blackvertex] (x) at (-1,2.25) {};
  \node[blackvertex] (y) at (1,2.25) {};
  \draw[thick,blue] (x)--(i) (y)--(j);
  \draw[thick,red,dashed] (x)--(j) (y)--(i);
  
  \node at (0,-1) {$(I,i,j)$};
\end{scope}

\begin{scope}[xshift=12.5cm]
  \node[blackvertex] (u0) at (0,0) {};
  \node[blackvertex] (v0) at (0,1.5) {};
  \node[blackvertex] (w0) at (0,3) {};
  \node[blackvertex] (x0) at (0,4.5) {};

  \node[blackvertex] (u1) at (2,0) {};
  \node[blackvertex] (v1) at (2,1.5) {};
  \node[blackvertex] (w1) at (2,3) {};
  \node[blackvertex] (x1) at (2,4.5) {};

  \draw[thick,black] (x0)--(x1) (x0)--(v1) (v0)--(x1);
  \draw[thick,black] (w0)--(w1) (u0)--(w1) (w0)--(u1);
  \draw[thick,black] (v0)--(v1) (v0) to[bend left=20] (x0) (v1) to[bend right=20] (x1);
  \draw[thick,black] (u0)--(u1) (u0) to[bend left=20] (w0) (u1) to[bend right=20] (w1);

  \draw[thick,black] (x0)--(v1)--(v0) to[bend left=20] (x0);
  \node at (1,-1) {$\rho(UC_4)^*$};

\end{scope}

\end{tikzpicture}
\end{center}
\caption{The unbalanced cycle $UC_4$, its switching graph $\rho(UC_4)$, the indicator $(I,i,j)$ and its resulting undirected graph $\rho(UC_4)^*$.}
\label{fig1}
\end{figure}

\section{Some NP-complete (non-signed) \PHOM{H} cases}
\label{sec:non-signed}

In this section, we prove that \PHOM{H} is NP-complete for two special
cases which were not known to be NP-complete.

\subsection{Squares of cycles}
\label{sec:cyclesquares}

We first deal with the case where $H$ is the square $C_t^2$ of a cycle
$C_t$. The proof is inspired by the proof that \PHOM{C_{2k+1}} is
NP-complete from~\cite{MS09}. Note that, $C_4^2=K_4$ and $C_5^2=K_5$,
and thus by the Four Colour Theorem \PHOM{C_t^2} is polynomial-time
solvable when $t=4,5$.

\begin{theorem}\label{thm:c2t}
For every $t\ge6$, \PHOM{C_t^2} is NP-complete.
\end{theorem}

\begin{proof} We will reduce from \PHOM{C_{2k+1}} for suitable values of $k$, which is NP-complete whenever $k\geq 1$~\cite{MS09}. The proof is split into different cases depending on the values of $t\bmod{3}$ and
  $t\bmod{4}$.

  If $t\equiv0\bmod{3}$, then the core of $C_t^2$ is
  $K_3$ and we are done since \PHOM{K_3} is NP-complete.

  Otherwise, $C_t^2$ is a core. Let $v_0,\ldots,v_{t-1}$ be its
  vertices and $v_iv_{i+1}$ and $v_iv_{i+2}$ be its edges (indices are
  taken modulo $t$).

  If $t\equiv2\bmod{4}$, then $C_t^2$ is planar, the set of edges
  $v_iv_{i+1}$ induces a cycle of length $t$, and the set of edges
  $v_iv_{i+2}$ induces two disjoint odd cycles of length $t/2$.  We
  reduce \PHOM{C_{t/2}} to \PHOM{C_t^2}.  Let $G$ be a planar graph
  and let $G'$ be the planar graph obtained from $G$ as follows.  For
  every edge $e$ of $G$, we add a copy of $C_t^2$ and we identify the
  edge $v_0v_2$ of this copy of $C_t^2$ with $e$. One can see that $G$ maps
  to $C_{t/2}$ if and only if $G$ maps to $C_t^2$, and we are done.

  If $t$ is odd, then $C_t^2$ is not planar, the set of edges
  $v_iv_{i+1}$ (resp. $v_iv_{i+2}$) induces an odd cycle and there
  exists no automorphism of $C_t^2$ that maps an edge $v_iv_{i+1}$ to
  an edge $v_iv_{i+2}$.  Consider the planar graph $H$ obtained from
  $C_t^2$ by removing the edge $v_0v_2$.  Notice that every
  homomorphism of $H$ to $C_t^2$ actually corresponds to an
  automorphism of $C_t^2$.  We reduce \PHOM{C_t} to \PHOM{C_t^2}.  Let
  $G$ be a planar graph and let $G'$ be the planar graph obtained from
  $G$ as follows.  For every edge $e$ of $G$, we add a copy of $H$ and
  we identify the edge $v_0v_1$ of this copy of $H$ with $e$.
  Again, $G$ maps to $C_t$ if and only if $G'$ maps to $C_t^2$ and we are done.

  If $t\equiv0\bmod{4}$, then both the set of edges $v_iv_{i+1}$ and
  the set of edges $v_iv_{i+2}$ induce a bipartite graph.  Thus, the
  kind of reductions above does not work.  We use instead a reduction
  similar to the one used for odd cycles in~\cite{MS09}. We reduce
  \PHOM{K_3} (that is, \textsc{Planar $3$-Colouring}) to \PHOM{C_t^2}.  Let us set $t=4k$. Consider the graph
  $H$ obtained from $C_t^2$ by removing the edges $v_{2k-1}v_{2k}$ and
  $v_{2k}v_{2k+2}$.  Notice that every homomorphism of $H$ to $C_t^2$
  that maps $v_0$ to $v_0$ also maps $v_{2k}$ to either $v_{2k-1}$,
  $v_{2k}$, or $v_{2k+1}$.

  Let $G$ be a planar graph and let $G'$ be the planar graph obtained
  from a planar embedding of $G$ as follows. For every face $f$ of
  $G$, we first place a new vertex $u_f$ inside the face $f$, and then
  for every vertex $w$ of $G$ incident to $f$, we add a copy of $H$
  and identify $v_0$ with $u_f$ and $v_{2k}$ with $w$. Every vertex in
  the subgraph $G$ of $G'$ is said to be \emph{old}. If $G$ is
  $3$-colourable, then we can map the subgraph $G$ of $G'$ to the
  triangle $v_{2k-1}v_{2k}v_{2k+1}$ of $C_t^2$. Then we can extend
  this $C_t^2$-colouring to $G'$ such that every vertex $u_f$ of $G'$
  maps to $v_0$.

  It remains to show that if $G'$ maps to $C_t^2$, then $G$ is
  $3$-colourable.  So we suppose for contradiction that there exists a
  planar graph $G$ such that $G'$ maps to $C_t^2$ and $G$ is not
  $3$-colourable.  Moreover, we choose such a graph $G$ with the
  minimum number of vertices.

  Let us first show that $G$ is a planar triangulation. Suppose to
  the contrary that $G$ contains a face $f$ of length at least $4$.
  Then $G'$ has a $C_t^2$-colouring that maps $u_f$ to $v_0$. So,
  every old vertex in $G'$ that corresponds to a vertex incident with
  $f$ in $G$ is mapped to a vertex in
  $\acc{v_{2k-1},v_{2k},v_{2k+1}}$.  Since $|f|\ge4$, two such old
  vertices $x$ and $y$ in $G'$ get the same colour.  Let $H$ be the
  planar graph obtained from $G$ by identifying $x$ and $y$ and
  removing multiple edges.  Notice that $H$ is not $3$-colourable,
  whereas the graph $H'$ obtained by applying our reduction to $H$
  maps to $C_t^2$. This contradicts the minimality of $G$ and thus $G$
  is a planar triangulation.

  Let $w$ be any old vertex. Since $G'$ maps to $C_t^2$, consider a
  $C_t^2$-colouring of $G'$ that maps $w$ to $v_0$.  Then every old
  vertex adjacent to $w$ maps to a vertex in
  $S=\acc{v_{-2},v_{-1},v_{1},v_{2}}$.  Notice that $S$ induces a path
  in $C_t^2$.  Since $G$ is a triangulation, the old vertices adjacent
  to $w$ induce a cycle $C$ which maps to $S$.  So $C$ maps to a
  bipartite graph and thus $C$ is bipartite.  This means that the
  length of $C$ is even and thus that the degree of $w$ in $G$ is
  even.  Thus, the degree of every vertex of $G$ is even, that is, $G$
  is an Eulerian planar triangulation.  This is a contradiction since
  every Eulerian planar triangulation is $3$-colourable (see Exercise
  9.6.2 in~\cite{BM76}).
\end{proof}

\subsection{Cubic circular cliques}\label{sec:circularcliques}
Recall that $K_{4t/(2t-1)}$ is the circular clique with vertex set
$\{k_0,\ldots,k_{4t-1}\}$ and such that $k_i$ is adjacent to $k_{i+2t-1}$,
$k_{i+2t}$, and $k_{i+2t+1}$ (indices being taken modulo $4t$). We now
use our result of Section~\ref{sec:cyclesquares} to show the
following.

\begin{theorem}\label{thm:cc}
For every $t\ge2$, \PHOM{K_{4t/(2t-1)}} is NP-complete.
\end{theorem}
\begin{proof}
  We reduce \PHOM{C_{4t}^2}, which is NP-complete by Theorem~\ref{thm:c2t}, to
  \PHOM{K_{4t/(2t-1)}}. For this, consider the $1$-edge-coloured
  indicator $(C_{2t+1},i,j)$ consisting of a cycle of length $2t+1$ on
  which $i$ and $j$ are at distance~$2$. Clearly, this indicator
  construction preserves planarity. Now, consider a homomorphism $f$
  from $(C_{2t+1},i,j)$ to $K_{4t/(2t-1)}$. By the symmetries of both
  graphs, we may assume that $f(i)=k_a$ for some
  $a\in\{0,\ldots,4k-1\}$. Note that the shortest odd cycles in
  $K_{4t/(2t-1)}$ are of length $2t+1$. Thus, $(C_{2t+1},i,j)$ must be
  mapped one-to-one. In fact, we must have
  $f(j)\in\{k_{a-2},k_{a-1},k_{a+1},k_{a+2}\}$ (where indices are
  taken modulo $4t$). Thus, the graph $K_{4t/(2t-1)}^*$ obtained from
  applying $(C_{2t+1},i,j)$ to $K_{4t/(2t-1)}$ is isomorphic to
  $C_{4t}^2$, and by Theorem~\ref{thm:indicator} the proof is
  complete.
%
%
%
\end{proof}

\section{Unbalanced even cycles of length at least $6$}\label{sec:uc2k}

We are now ready to prove that unbalanced even cycles of length at
least~$6$ define an NP-complete s-homomorphism problem for planar
graphs.

\begin{theorem}\label{thm:uc2k}
For every $k\ge3$, \Pshom{UC_{2k}} is NP-complete.
\end{theorem}
\begin{proof}
  We will equivalently show that \PHOM{\rho(UC_{2k})} is NP-complete.

  Let $u_0,\ldots,u_{4k-1}$ be the vertices of $\rho(UC_{2k})$. The
  positive edges of $\rho(UC_{2k})$ are $u_iu_{i+1}$ and the negative
  edges are $u_iu_{i+2k-1}$.

  Consider the indicator $(I,i,j)$ of Figure~\ref{fig1}. Clearly, it
  preserves planarity. Now, consider a homomorphism $f$ of $(I,i,j)$
  to $\rho(UC_{2k})$. By the symmetries of both graphs, we may assume
  that $f(i)=u_a$ for some $a\in\{0,\ldots,4k-1\}$. Then, we must have
  $f(j)\in\{u_{a+2k-2},u_{a+2k},u_{a+2k+2}\}$ (indices taken modulo
  $4k$).

  Thus, the graph $\rho(UC_{2k})^*$ obtained from applying $(I,i,j)$
  to $\rho(UC_{2k})$ is cubic. Notice that $\rho(UC_{2k})$ is
  bipartite, and that $(I,i,j)$ is bipartite too with $i$ and $j$ in
  the same part. It follows that $\rho(UC_{2k})^*$ will consist of at
  least two connected components. See Figure~\ref{fig:UV6-UC8} for a
  picture when $k=3,4$. We now distinguish two cases to determine
  $\rho(UC_{2k})^*$.

\medskip

\emph{Case 1: $k$ is odd.} In this case, $\rho(UC_{2k})^*$ contains a
copy of the cycle $C_{k}=\{c_0,\ldots,c_{k-1}\}$, where $c_0=u_0$ and
$c_{i+1}=u_{i+2k+2}$ when $i$ is even and $c_{i+1}=u_{i-(2k-2)}$ when
$i$ is odd (thus, $c_k=u_{2k-2}$). Furthermore, we claim that $C_k$ is
the core of $\rho(UC_{2k})^*$: indeed, consider the $k$ sets of
vertices $U_j=\{u_j,\ldots,u_{j+3}\}$ for $j=0\bmod 4$ and
$0\leq j<4k$. For $u_i\in U_j$, let $f(u_i)=c_j$: $f$ is a
homomorphism from $\rho(UC_{2k})^*$ to its subgraph $C_k$.

\medskip

\emph{Case 2: $k$ is even.} In this case, letting $k=2t$ with
$t\geq 2$, we have that $\rho(UC_{2k})^*$ consists of two copies of
the circular clique $K_{4t/(2t-1)}$: one on vertex set
$\{u_i~|~i=0\bmod 2\}$ and the other on vertex set
$\{u_i~|~i=1\bmod 2\}$. Thus, the core of $\rho(UC_{2k})^*$ is
$K_{4t/(2t-1)}$.

\medskip

In both cases, we have \PHOM{\rho(UC_{2k})^*} NP-complete:
by~\cite{HNT06,MS09} when $k$ is odd and by Theorem~\ref{thm:cc} when
$k$ is even. We thus apply Theorem~\ref{thm:indicator} to obtain a
reduction from \PHOM{\rho(UC_{2k})^*} to \PHOM{\rho(UC_{2k})}: this
completes the proof.
\end{proof}

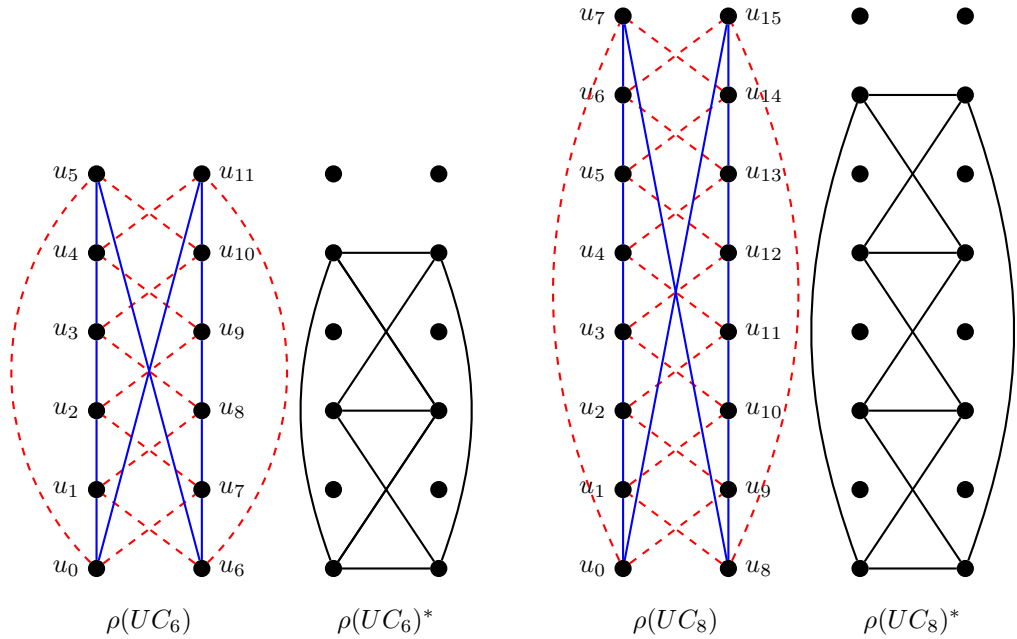
\begin{figure}[!htpb]
\begin{center}
\begin{tikzpicture}[every loop/.style={},scale=0.7]
\begin{scope}[xshift=0cm]
  \node[blackvertex] (u0) at (0,0) {};
  \draw (u0) node[left=1mm] {$u_0$};
  \node[blackvertex] (u1) at (0,1.5) {};
  \draw (u1) node[left=1mm] {$u_1$};
  \node[blackvertex] (u2) at (0,3) {};
  \draw (u2) node[left=1mm] {$u_2$};
  \node[blackvertex] (u3) at (0,4.5) {};
  \draw (u3) node[left=1mm] {$u_3$};
  \node[blackvertex] (u4) at (0,6) {};
  \draw (u4) node[left=1mm] {$u_4$};
  \node[blackvertex] (u5) at (0,7.5) {};
  \draw (u5) node[left=1mm] {$u_5$};
  \draw[thick,blue] (u0)--(u1)--(u2)--(u3)--(u4)--(u5);
  \draw[thick,red,dashed] (u0) to[bend left=45] (u5);

  \node[blackvertex] (u6) at (2,0) {};
  \draw (u6) node[right=1mm] {$u_6$};
  \node[blackvertex] (u7) at (2,1.5) {};
  \draw (u7) node[right=1mm] {$u_7$};
  \node[blackvertex] (u8) at (2,3) {};
  \draw (u8) node[right=1mm] {$u_8$};
  \node[blackvertex] (u9) at (2,4.5) {};
  \draw (u9) node[right=1mm] {$u_9$};
  \node[blackvertex] (u10) at (2,6) {};
  \draw (u10) node[right=1mm] {$u_{10}$};
  \node[blackvertex] (u11) at (2,7.5) {};
  \draw (u11) node[right=1mm] {$u_{11}$};
  \draw[thick,blue] (u6)--(u7)--(u8)--(u9)--(u10)--(u11);
  \draw[thick,red,dashed] (u6) to[bend right=45] (u11);
  
  \draw[thick,red,dashed] (u0)--(u7)--(u2)--(u9)--(u4)--(u11) (u6)--(u1)--(u8)--(u3)--(u10)--(u5);
  \draw[thick,blue] (u0)--(u11) (u5)--(u6);
  
  \node at (1,-1) {$\rho(UC_6)$};
\end{scope}

\begin{scope}[xshift=4.5cm]
  \node[blackvertex] (u0) at (0,0) {};
  \node[blackvertex] (u1) at (0,1.5) {};
  \node[blackvertex] (u2) at (0,3) {};
  \node[blackvertex] (u3) at (0,4.5) {};
  \node[blackvertex] (u4) at (0,6) {};
  \node[blackvertex] (u5) at (0,7.5) {};

  \node[blackvertex] (u6) at (2,0) {};
  \node[blackvertex] (u7) at (2,1.5) {};
  \node[blackvertex] (u8) at (2,3) {};
  \node[blackvertex] (u9) at (2,4.5) {};
  \node[blackvertex] (u10) at (2,6) {};
  \node[blackvertex] (u11) at (2,7.5) {};

  \draw[thick,black] (u0)--(u6)--(u2)--(u8)--(u0) (u2)--(u10)--(u4)--(u8);
  \draw[thick,black] (u0)--(u8)--(u4);
  \draw[thick,black] (u0) to[bend left=20] (u4);
  \draw[thick,black] (u6) to[bend right=20] (u10);
  \node at (1,-1) {$\rho(UC_6)^*$};
\end{scope}

\begin{scope}[xshift=10cm]
  \node[blackvertex] (u0) at (0,0) {};
  \draw (u0) node[left=1mm] {$u_0$};
  \node[blackvertex] (u1) at (0,1.5) {};
  \draw (u1) node[left=1mm] {$u_1$};
  \node[blackvertex] (u2) at (0,3) {};
  \draw (u2) node[left=1mm] {$u_2$};
  \node[blackvertex] (u3) at (0,4.5) {};
  \draw (u3) node[left=1mm] {$u_3$};
  \node[blackvertex] (u4) at (0,6) {};
  \draw (u4) node[left=1mm] {$u_4$};
  \node[blackvertex] (u5) at (0,7.5) {};
  \draw (u5) node[left=1mm] {$u_5$};
  \node[blackvertex] (u6) at (0,9) {};
  \draw (u6) node[left=1mm] {$u_6$};
  \node[blackvertex] (u7) at (0,10.5) {};
  \draw (u7) node[left=1mm] {$u_7$};
  \draw[thick,blue] (u0)--(u1)--(u2)--(u3)--(u4)--(u5)--(u6)--(u7);
  \draw[thick,red,dashed] (u0) to[bend left=25] (u7);

  \node[blackvertex] (u8) at (2,0) {};
  \draw (u8) node[right=1mm] {$u_8$};
  \node[blackvertex] (u9) at (2,1.5) {};
  \draw (u9) node[right=1mm] {$u_9$};
  \node[blackvertex] (u10) at (2,3) {};
  \draw (u10) node[right=1mm] {$u_{10}$};
  \node[blackvertex] (u11) at (2,4.5) {};
  \draw (u11) node[right=1mm] {$u_{11}$};
  \node[blackvertex] (u12) at (2,6) {};
  \draw (u12) node[right=1mm] {$u_{12}$};
  \node[blackvertex] (u13) at (2,7.5) {};
  \draw (u13) node[right=1mm] {$u_{13}$};
  \node[blackvertex] (u14) at (2,9) {};
  \draw (u14) node[right=1mm] {$u_{14}$};
  \node[blackvertex] (u15) at (2,10.5) {};
  \draw (u15) node[right=1mm] {$u_{15}$};
  
  \draw[thick,blue] (u8)--(u9)--(u10)--(u11)--(u12)--(u13)--(u14)--(u15);
  \draw[thick,red,dashed] (u8) to[bend right=25] (u15);
  
  \draw[thick,red,dashed] (u0)--(u9)--(u2)--(u11)--(u4)--(u13)--(u6)--(u15) (u8)--(u1)--(u10)--(u3)--(u12)--(u5)--(u14)--(u7);
  \draw[thick,blue] (u0)--(u15) (u7)--(u8);
  
  \node at (1,-1) {$\rho(UC_8)$};
\end{scope}

\begin{scope}[xshift=14.5cm]
  \node[blackvertex] (u0) at (0,0) {};
  \node[blackvertex] (u1) at (0,1.5) {};
  \node[blackvertex] (u2) at (0,3) {};
  \node[blackvertex] (u3) at (0,4.5) {};
  \node[blackvertex] (u4) at (0,6) {};
  \node[blackvertex] (u5) at (0,7.5) {};
  \node[blackvertex] (u6) at (0,9) {};
  \node[blackvertex] (u7) at (0,10.5) {};

  \node[blackvertex] (u8) at (2,0) {};
  \node[blackvertex] (u9) at (2,1.5) {};
  \node[blackvertex] (u10) at (2,3) {};
  \node[blackvertex] (u11) at (2,4.5) {};
  \node[blackvertex] (u12) at (2,6) {};
  \node[blackvertex] (u13) at (2,7.5) {};
  \node[blackvertex] (u14) at (2,9) {};
  \node[blackvertex] (u15) at (2,10.5) {};

  \draw[thick,black] (u0)--(u8)--(u2)--(u10) (u0) (u2)--(u12)--(u4) (u6)--(u12);
  \draw[thick,black] (u0)--(u10)--(u4)--(u14)--(u6);
  \draw[thick,black] (u0) to[bend left=20] (u6);
  \draw[thick,black] (u8) to[bend right=20] (u14);
  \node at (1,-1) {$\rho(UC_8)^*$};
\end{scope}

\end{tikzpicture}
\end{center}
\caption{The switching graphs $\rho(UC_6)$ and $\rho(UC_8)$, and the undirected graphs $\rho(UC_6)^*$ and $\rho(UC_8)^*$ resulting from the application of the indicator $(I,i,j)$ from Figure~\ref{fig1} (in which we only draw one of the two isomorphic components). The core of $\rho(UC_6)^*$ is $K_3$ and the core of $\rho(UC_8)^*$ is $K_{8/3}$, the Wagner graph.}
\label{fig:UV6-UC8}
\end{figure}

\section{Unbalanced cycles of length 4}\label{sec:uc4}

The proof of Theorem~\ref{thm:uc2k} fails to show that \PHOM{UC_4} is
NP-complete since $\rho(UC_4^*)$ is $K_4$ and thus this would require
\PHOM{K_4} to be NP-complete, which is false by the Four Colour
Theorem. However, we give an ad-hoc reduction in this section.

\begin{theorem}\label{thm:uc4}
\Pshom{UC_4} is NP-complete.
\end{theorem}

As before, we show equivalently that \PHOM{\rho(UC_4)} is
NP-complete. It is easy to see that \PHOM{\rho(UC_4)} is in NP. We
thus focus on proving NP-hardness. Recall that $\rho(UC_4)$ is the
circulant graph with vertices $u_0,\ldots,u_7$, positive edges
$u_iu_{i+1}$ and negative edges $u_iu_{i+3}$ (it is depicted in
Figure~\ref{fig1}, see also Figure~\ref{fig:rho(UC_4)} for a symmetric drawing). We reduce \PHOM{K_3} (that is, \textsc{Planar $3$-Colouring}), which is
NP-complete~\cite{GJS76}, to \PHOM{\rho(UC_4)}. In the following, $G$
is an instance of \PHOM{K_3}. Our goal is to construct a planar signed
graph $H$ such that $G$ is $3$-colourable if and only if $H$ is a
positive instance of \PHOM{\rho(UC_4)}. We assume that $G$ is
connected, otherwise, we apply our construction to each connected
component.

\begin{figure}[!htpb]
\begin{center}
\scalebox{1.0}{\begin{tikzpicture}[every loop/.style={},scale=0.7]
\node[blackvertex](u0) at (0:2) {};
\node at (0:2.5) {$u_0$};
\node[blackvertex](u1) at (45:2) {};
\node at (45:2.5) {$u_1$};
\node[blackvertex](u2) at (90:2) {};
\node at (90:2.5) {$u_2$};
\node[blackvertex](u3) at (135:2) {};
\node at (135:2.5) {$u_3$};
\node[blackvertex](u4) at (180:2) {};
\node at (180:2.5) {$u_4$};
\node[blackvertex](u5) at (225:2) {};
\node at (225:2.5) {$u_5$};
\node[blackvertex](u6) at (270:2) {};
\node at (270:2.5) {$u_6$};
\node[blackvertex](u7) at (315:2) {};
\node at (315:2.5) {$u_7$};

\draw[thick,red,dashed] (u0)--(u3)--(u6)--(u1)--(u4)--(u7)--(u2)--(u5)--(u0);

\draw[thick,blue] (u0)--(u1)--(u2)--(u3)--(u4)--(u5)--(u6)--(u7)--(u0);

\end{tikzpicture}}
\end{center}
\caption{The signed graph $\rho(UC_4)$.}
\label{fig:rho(UC_4)}
\end{figure}
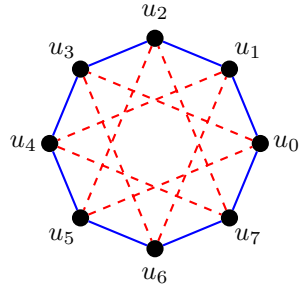

\subsection{Outline of the reduction}
We first present the generic ideas of our reduction. Note that the
graph $\rho(UC_4)$ is bipartite. Therefore, $H$ will be bipartite,
otherwise there is no homomorphism from $H$ to $\rho(UC_4)$. In our
construction of $H$, the important vertices belong to the same
partite set. All of the vertices that will be named in our
gadgets will always belong to that partite set. By the symmetry of
$\rho(UC_4)$, we can assume that, in every homomorphism from one of
our graphs to $\rho(UC_4)$, all of the named vertices will always be
mapped to some $u_i$ with $i = 0 \bmod 2$.

For each vertex $v$ of degree $d$ in $G$, we will create a gadget $G_v$ in $H$
with $2d$ special vertices $v_0,\cdots v_{2d-1}$ (called \emph{ports})
such that there is an embedding of that gadget in the plane where
$v_0,\cdots v_{2d-1}$ are on a facial trail in that order. We want
that any homomorphism $\varphi$ from $H$ to $\rho(UC_4)$ satisfies:
\begin{enumerate}
\item \label{cond:3col} $\varphi(v_0)\neq\varphi(v_1)$
\item \label{cond:period} for $i=0,1$,
  $\varphi(v_i)=\varphi(v_{(i\bmod 2)})$.
\end{enumerate}

We will describe later how to enforce these conditions. Assuming that
they are satisfied, let $\varphi$ be any homomorphism from $H$ to
$\rho(UC_4)$. Let $\varphi(v_1)=u_l$ and $\varphi(v_0)=u_k$. The
difference $l-k \bmod 8$ will represent the \emph{colour} of $v$ in a
valid 3-colouring of $G$. For the sake of readability, we identify
$u_k$ with its index $k$, so that we can read the colour of $v$ by the
operation $\varphi(v_1)-\varphi(v_0)\bmod 8$. We shall call
$\varphi(v_0)$ the \emph{ground} of $v$ and
$\varphi(v_1)-\varphi(v_0)$ its \emph{colour}. Note that
Condition~\ref{cond:3col} ensures that there are only three possible
colours for $v$: $2$, $4$, and $6$. Condition~\ref{cond:period}
asserts that the colour of $v$ propagates $d$ times throughout the
vertex gadget as
$\varphi(v_{2i+1})-\varphi(v_{2i})\bmod 8, i\in\{0, \ldots, d-1\}$,
allowing us to use any pair $(v_{2i},v_{2i+1})$ for retrieving it.

Recall that we want $H$ to be a positive instance of \PHOM{\rho(UC_4)}
if and only if $G$ is $3$-colourable. Therefore, we also want to
ensure that any homomorphism $H\to \rho(UC_4)$ assigns different
colours to each pair of adjacent vertices in $G$. We thus construct
$H$ such that the following condition is satisfied:
\begin{enumerate}
\setcounter{enumi}{2}
\item\label{cond:proper} For each homomorphism
  $\varphi: H\to \rho(UC_4)$, two adjacent vertices in $G$ receive the
  same ground, and different colours.
\end{enumerate}

Note that if we manage to construct $H$ such that the three conditions
hold, then it is easy to recover a proper $3$-colouring of $G$ from
any homomorphism $H\to \rho(UC_4)$. In the following subsections, we
describe our gadgets, and we prove that the previous conditions hold.

\subsection{Construction of the gadgets}

Tu build our vertex gadget, we need to start with several smaller gadgets. We start with a first construction allowing to make copies of a vertex
that are mapped to the same image under any homomorphism, called \emph{copy gadget} (see Figure~\ref{fig:cross_aux}).

  \begin{figure}[!ht]
    \centering
    \begin{tikzpicture}[thick]
      \node[smallblack][label=left:{$x_1$}] (01) at  (0,1){};
      \node[smallblack] (02) at (0,2) {};
      \node[smallblack][label=below:{$y_2$}] (10) at (1,0){};
      \node[smallblack] (11) at (1,1){};
      \node[smallblack][label=above:{$y_1$}] (12) at (1,2){};
      \node[smallblack] (20) at (2,0){};
      \node[smallblack][label=right:{$x_2$}] (21) at (2,1){};
      \node[smallblack] (22) at  (2,2){};
      \draw[blue] (01) -- (02) -- (22) -- (20) -- (10);
      \draw[blue] (10) -- (12);
      \draw[red,dashed] (01) -- (11) -- (21);
    
      \tikzset{xshift=5cm,yshift=2cm}
      \draw (0,0) -- (2,0) -- (2,-2) -- (0,-2) -- (0,0);
      \draw (1,0.5) -- (1,0);
      \draw (-0.5,-1) -- (0,-1);
      \draw (2.5,-1) -- (2,-1);
      \draw (1,-2.5) -- (1,-2);
      \node at (1,-0.2) {$y_1$};
      \node at (0.2,-1) {$x_1$};
      \node at (1.8,-1) {$x_2$};
      \node at (1,-1.8) {$y_2$};
      \node at (1,-1) {\Huge $=$};
    \end{tikzpicture}
    \caption{Copy gadget and its schematic representation.}
    \label{fig:cross_aux}
  \end{figure}
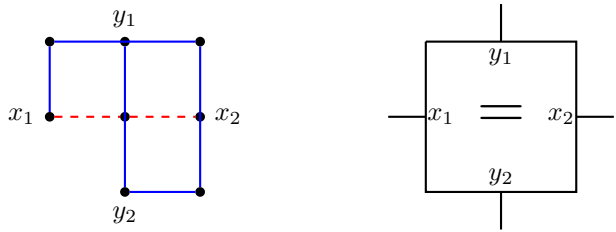

\begin{lemma}  \label{lem:cross_aux}\label{lemma:copy}
  Let $\varphi$ be a homomorphism from the copy gadget from
  Figure~\ref{fig:cross_aux} to $\rho(UC_4)$. 

  Then, $\varphi(x_1)=\varphi(x_2)$, $\varphi(y_1)=\varphi(y_2)$ and
  $\varphi(y_1)-\varphi(x_1)\in \{\pm 2\}$.

  Conversely, if we partially fix $\varphi$ from $\{x_1,x_2,y_1,y_2\}$ to
  $\{u_0,u_2,u_4,u_6\}$ such that $\varphi(x_1)=\varphi(x_2)$,
  $\varphi(y_1)=\varphi(y_2)$ and $\varphi(y_1)-\varphi(x_1)\in \{\pm 2\}$,
  one can extend $\varphi$ to a homomorphism from the copy gadget to $\rho(UC_4)$.
\end{lemma}
\begin{proof}
  The first thing to notice is that in a homomorphism from a copy of
  $UC_4$ to $\rho(UC_4)$, each vertex maps to a distinct vertex, and
  the only ways to do it are shown in Figure~\ref{fig:UC4}.

  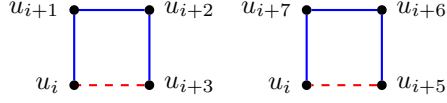
\begin{figure}[!ht]
    \centering
    \begin{tikzpicture}[thick]
      \node[smallblack][label=left:{$u_{i+1}$}] (01) at  (0,1){};
      \node[smallblack][label=right:{$u_{i+3}$}] (10) at (1,0){};
      \node[smallblack][label=right:{$u_{i+2}$}] (11) at (1,1){};
      \node[smallblack][label=left:{$u_i$}] (00) at  (0,0){};
      \draw[blue] (00) -- (01) -- (11) -- (10);
      \draw[red,dashed] (10) -- (00);
    \end{tikzpicture}
    \begin{tikzpicture}[thick]
      \node[smallblack][label=left:{$u_{i+7}$}] (01) at  (0,1){};
      \node[smallblack][label=right:{$u_{i+5}$}] (10) at (1,0){};
      \node[smallblack][label=right:{$u_{i+6}$}] (11) at (1,1){};
      \node[smallblack][label=left:{$u_i$}] (00) at  (0,0){};
      \draw[blue] (00) -- (01) -- (11) -- (10);
      \draw[red,dashed] (10) -- (00);
    \end{tikzpicture}
    \caption{The only possible homomorphisms for a copy of $UC_4$. The indices are taken modulo~$8$.}
    \label{fig:UC4}
  \end{figure}

  Given the image of a vertex, there are two possibilities for the
  remainder of the copy of $UC_4$, and given the images of two vertices,
  there is at most one way to complete the homomorphism. In a
  homomorphism from the gadget to $\rho(UC_4)$, assume without loss of generality that $x_1$ maps to
  $u_0$. The two possibilities to map the vertices of the copy of
  $UC_4$ containing $x_1$ lead to $y_1$ being mapped either to $u_2$
  or $u_6$. Then, the mapping of the remainder of the vertices is forced,
  leading to $x_2$ being mapped to $u_0$ and $y_2$ to the same vertex
  as $y_1$. The symmetries of $UC_4$ complete the proof of the lemma.
\end{proof}

Given a vertex $v$ in $G$, we want to use the copy gadget to ensure that
Condition~\ref{cond:period} holds in the vertex gadget $G_v$, by identifying $x_1$ with
$v_{2i}$ and $x_2$ with $v_{2i+2}$ for $i=0,\ldots,d-1$ (indices are
taken modulo $2d$).  However, we also need to have a copy of the
ground between each $x_i$.

To this end, we have to design a crossing-type gadget. Observe that the
copy gadget allows to cross two images as soon as they differ by
$\pm 2$.

We thus need to find a gadget that can handle the case where the
difference is~$4$. To this end, we introduce the \emph{split gadget} from Figure~\ref{fig:split}, which
allows to encode a difference of $\pm 2$ or $\pm 4$ using two differences
of $\pm 2$.

  \begin{figure}[!ht]
    \centering
    \begin{tikzpicture}[thick]
      \node[smallblack,label=above:{$x_1$}] (i) at (0,0) {};
      \node[smallblack,label=left:{$x_2$}] (o1) at (0,-2) {};
      \node[smallblack,label=above:{$y_2$}] (o2) at (2,0) {};
      \node[smallblack,label=left:{$g_1$}] (g1) at (-2,0) {};
      \node[smallblack,label=right:{$g_2$}] (g2) at (1,-3){};
      \node[smallblack,label=below:{$g_3$}] (g3) at (3,-1) {};
      \node[smallblack] (ii) at (1,-1){};
      \node[smallblack] (io1) at (0,-1){};
      \node[smallblack] (io2) at (1,0) {};
      \node[smallblack] (o2g3) at (2,-1) {};
      \node[smallblack] (o2g3b) at (3,0){};
      \node[smallblack] (o1g2) at (0,-3){};
      \node[smallblack] (o1g2b) at (1,-2) {};
      \node[smallblack] (ig1) at (-1,0) {};
      \draw[blue] (ig1) -- (i) -- (io1) -- (o1) -- (o1g2) -- (g2);
      \draw[blue] (o1) -- (o1g2b);
      \draw[blue] (o2g3b) -- (o2) -- (o2g3) -- (g3);
      \draw[blue] (i) -- (io2) -- (o2);
      \draw[blue] (io1) -- (ii);
      \draw[red,dashed] (o1g2b) -- (g2);
      \draw[red,dashed] (o2g3b) -- (g3);
      \draw[red,dashed] (io2) -- (ii);
      \draw[red,dashed] (g1) -- (ig1);


      
      \tikzset{xshift=5cm}
      \draw (0,0) -- (2,0) -- (2,-2) -- (0,-2) -- (0,0);
      \draw (0.66,0.5) -- (0.66,0);
      \draw (1.33,0.5) -- (1.33,0);
      \draw (-0.5,-0.66) -- (0,-0.66);
      \draw (-0.5,-1.33) -- (0,-1.33);
      \draw (2.5,-0.66) -- (2,-0.66);
      \draw (2.5,-1.33) -- (2,-1.33);
      \node at (0.66,-0.2) {$g_1$};
      \node at (1.33,-0.2) {$x_1$};
      \node at (0.2,-0.66) {$x_2$};
      \node at (0.2,-1.33) {$g_2$};
      \node at (1.8,-0.66) {$y_2$};
      \node at (1.8,-1.33) {$g_3$};
      \node at (1,-1) {\Huge $ \curlywedge$};
    \end{tikzpicture}
    \caption{Split gadget and its schematic representation.}
    \label{fig:split}
  \end{figure}
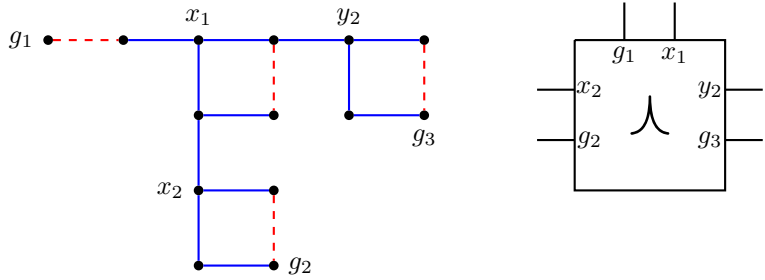

  \begin{lemma}\label{lemma:split}
     Let $\varphi$ be a homomorphism from the split gadget from
  Figure~\ref{fig:split} to $\rho(UC_4)$ such that
  $\varphi(g_1)=\varphi(g_2)=\varphi(g_3)$.

  Then, $\varphi(x_2)-\varphi(g_2)$ and $\varphi(y_2)-\varphi(g_3)$ lie
  in $\{\pm 2\}$. Moreover, $\varphi(x_1)-\varphi(g_1)=4$ if and only
  if $\varphi(x_2)\neq \varphi(y_2)$, and
  $\varphi(x_1)-\varphi(g_1)\in\{\pm 2\}$ if and only if
  $\varphi(x_2)=\varphi(y_2)=\varphi(x_1)$.

  Conversely, if we partially fix $\varphi$ from $\{x_2,y_2,x_1,g_1,g_2,g_3\}$
  to $\{u_0,u_2,u_4,u_6\}$ such that
  $\varphi(g_1)=\varphi(g_2)=\varphi(g_3)$,
  $\varphi(x_2)-\varphi(g_2)$ and $\varphi(y_2)-\varphi(g_3)$ lie in
  $\{\pm 2\}$, $\varphi(x_1)-\varphi(g_1)=4$ if and only if
  $\varphi(x_2)\neq \varphi(y_2)$, and
  $\varphi(x_1)-\varphi(g_1)\in\{\pm 2\}$ if and only if
  $\varphi(x_2)=\varphi(y_2)=\varphi(x_1)$, then one can extend $\varphi$ to
  a homomorphism from the split gadget to
  $\rho(UC_4)$.
\end{lemma}  

\begin{proof}
  Consider a homomorphism $\varphi$ from the gadget to $\rho(UC_4)$.
  Say that $g_1$, $g_2$, and $g_3$ are mapped to $u_0$. Similarly to
  the proof of Lemma~\ref{lem:cross_aux}, $x_2$ and $y_2$ can each only be mapped to $u_2$
  or $u_6$. Furthermore, $x_1$ can only be mapped to $u_2$, $u_4$ or
  $u_6$. Now, if $x_1$ is mapped to $u_4$, then its two neighbours in
  its copy of $UC_4$ are mapped one to $u_3$ and one to $u_5$, forcing
  $x_2$ and $y_2$ to be mapped one to $u_2$ and the other to $u_6$.
  If $x_1$ is mapped to $u_2$, then its two neighbours in its copy of
  $UC_4$ are mapped one to $u_1$ and one to $u_3$, forcing $x_2$ and
  $y_2$ to be mapped either both to $u_2$ or both to $u_6$. The case
  when $x_1$ is mapped to $u_6$ is similar. If $g_1$, $g_2$, and
  $g_3$ are mapped to another $u_i$, then the symmetries of
  $\rho(UC_4)$ yield the result.

  To prove the converse, one can check that every time we specified
  that something was forced above, there actually existed a
  homomorphism that verified the conditions.
\end{proof}

Note that the split gadget can be used in two ways: for splitting an image
in $\{\pm 2, \pm 4\}$ into two images in $\{\pm 2\}$, but also ``backwards''
for combining two images in $\{\pm 2\}$ into an image in
$\{\pm 2, \pm 4\}$. Thus, by combining the copy gadget and the split gadget, we obtain our crossing gadget, that allows a crossing between two images
that differ from $\pm 2$ or $\pm 4$, see Figure~\ref{fig:cross}. It is clear that the result is planar. Notice that the six vertices towards the center (the lower vertices of the upper split gadget, the upper vertices of the lower split gadget, the right vertex of the left copy gadget, and the left vertex of the right copy gadget) are all identified into one vertex.

  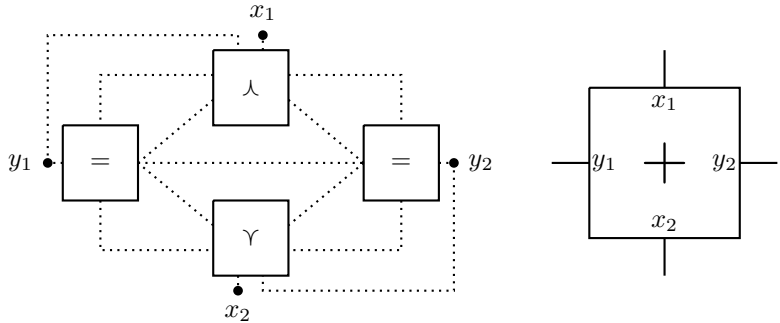
\begin{figure}[!ht]
    \centering
    \begin{tikzpicture}[thick]
      \draw (0,0) -- (1,0) -- (1,-1) -- (0,-1) -- (0,0);
      \node at (0.5,-0.5) {$\curlywedge$};

      \draw (0,-2) -- (1,-2) -- (1,-3) -- (0,-3) -- (0,-2);
      \node at (0.5,-2.5) {$\curlyvee$};
      
      \draw (-2,-1) -- (-1,-1) -- (-1,-2) -- (-2,-2) -- (-2,-1);
      \node at (-1.5,-1.5) {$=$};
      
      \draw (2,-1) -- (3,-1) -- (3,-2) -- (2,-2) -- (2,-1);
      \node at (2.5,-1.5) {$=$};

      \draw[dotted] (0.66,0) -- (0.66,0.2);
      \draw[dotted] (0.33,-3) -- (0.33,-3.2);
      \draw[dotted] (-2,-1.5) -- (-2.2,-1.5);
      \draw[dotted] (3,-1.5) -- (3.2,-1.5);
      \draw[dotted] (-2.2,-1.5) -- (-2.2,0.2) -- (0.33,0.2) -- (0.33,0);
      \draw[dotted] (3.2,-1.5) -- (3.2,-3.2) -- (0.66,-3.2) -- (0.66,-3);

      \draw[dotted] (0,-0.33) -| (-1.5,-1);
      \draw[dotted] (1,-0.33) -| (2.5,-1);

      \draw[dotted] (0,-0.66) -- (-1,-1.5);
      \draw[dotted] (0,-2.33) -- (-1,-1.5);
      \draw[dotted] (2,-1.5) -- (1,-0.66);
      \draw[dotted] (2,-1.5) -- (1,-2.33);
      \draw[dotted] (-1,-1.5) to (2,-1.5);
      \draw[dotted] (-1.5,-2) |- (0,-2.66);
      \draw[dotted] (2.5,-2) |- (1,-2.66);

      \node[smallblack,label=above:{$x_1$}] at (0.66,0.2) {};
      \node[smallblack,label=below:{$x_2$}] at (0.33,-3.2) {};
      \node[smallblack,label=left:{$y_1$}] at (-2.2,-1.5) {};
      \node[smallblack,label=right:{$y_2$}] at (3.2,-1.5) {};

      \tikzset{xshift=5cm,yshift=-0.5cm}
      \draw (0,0) -- (2,0) -- (2,-2) -- (0,-2) -- (0,0);
      \draw (1,0.5) -- (1,0);
      \draw (-0.5,-1) -- (0,-1);
      \draw (2.5,-1) -- (2,-1);
      \draw (1,-2.5) -- (1,-2);
      \node at (1,-0.2) {$x_1$};
      \node at (0.2,-1) {$y_1$};
      \node at (1.8,-1) {$y_2$};
      \node at (1,-1.8) {$x_2$};
      \node at (1,-1) {\Huge $ +$};
    \end{tikzpicture}
    \caption{Construction of the crossing gadget (dotted links mean identification of
      vertices) and its schematic representation.}
    \label{fig:cross}
  \end{figure}

\begin{lemma}
  \label{lem:cross}
  Let $\varphi$ be a homomorphism from the crossing gadget depicted in
  Figure~\ref{fig:cross} to $\rho(UC_4)$.
  
  Then, $\varphi(x_1)=\varphi(x_2)$, $\varphi(y_1)=\varphi(y_2)$ and
  $\varphi(y_1)-\varphi(x_1)\in\{\pm 2, \pm 4\}$.

  Conversely, if we partially fix $\varphi$ from $\{x_1,x_2,y_1,y_2\}$ to
  $\{u_0,u_2,u_4,u_6\}$ such that $\varphi(x_1)=\varphi(x_2)$,
  $\varphi(y_1)=\varphi(y_2)$ and
  $\varphi(y_1)-\varphi(x_1)\in\{\pm 2, \pm 4\}$, then one can extend $\varphi$ to
  a homomorphism from the crossing gadget to
  $\rho(UC_4)$.
\end{lemma}

\begin{proof}
  The proof is a direct application of Lemmas~\ref{lemma:copy} and~\ref{lemma:split}. Consider a
  homomorphism from the gadget to $\rho(UC_4)$. The vertices $y_1$ and
  $y_2$ must be mapped to the same vertex due to the copy gadgets. Say $y_1$ and $y_2$ are mapped to
  $u_0$. Notice that in addition, due to the copy gadgets and the identifications of vertices, all vertices of type $g_i$ from the two split gadgets are mapped to $u_0$ as well. Thus, Lemma~\ref{lemma:split} can be applied to the two split gadgets.

By Lemma~\ref{lemma:split}, $x_1$ is mapped to $u_2$, $u_4$, or $u_6$. If $x_1$ is
  mapped to $u_2$ or $u_6$ (say, $u_2$), then by Lemma~\ref{lemma:split}, the remaining two labeled vertices (other than the $g_i$'s) of the upper split
  gadget are also mapped to $u_2$. So are their ``copies'' in the two copy gadgets, and by applying Lemma~\ref{lemma:split} to the lower split gadget, so is
  $x_2$, as desired.

  If $x_1$ is mapped
  to $u_4$, then by Lemma~\ref{lemma:split} the remaining two vertices of the split gadget are
  mapped one to $u_2$ and one to $u_6$; thus, again, so are their copies in
  the copy gadgets, and thus again by Lemma~\ref{lemma:split} $x_2$ is mapped to $u_4$.

  The existence of a homomorphism extending any good mapping of $x_1$,
  $x_2$, $y_1$, and $y_2$ to $\{u_0,u_2,u_4,u_6\}$ can be checked
  similarly.
\end{proof}

The whole vertex gadget is now created by gluing crossing gadgets, as
shown in Figure~\ref{fig:wholev}. Due to Lemma~\ref{lem:cross}, we
obtain that Conditions~\ref{cond:3col} and~\ref{cond:period} are
satisfied.

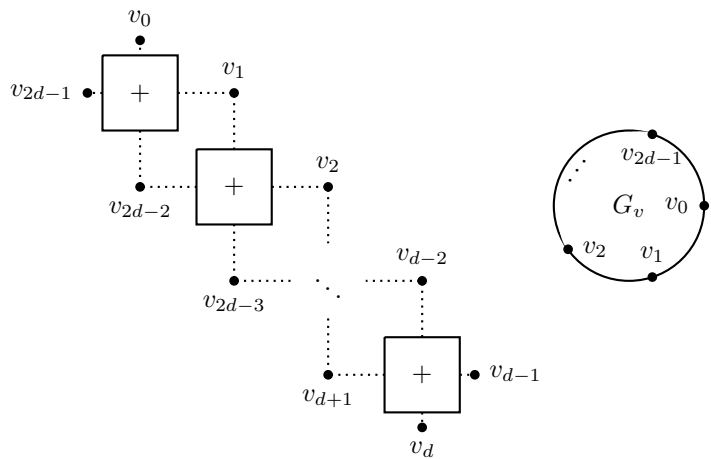
\begin{figure}[!ht]
    \centering
    \begin{tikzpicture}[thick]
      \draw (0,0) -- (1,0) -- (1,-1) -- (0,-1) -- (0,0);
      \node at (0.5,-0.5) {$+$};

      \draw (1.25,-1.25) -- (2.25,-1.25) -- (2.25,-2.25) -- (1.25,-2.25) -- (1.25,-1.25);
      \node at (1.75,-1.75) {$+$};

      \draw (3.75,-3.75) -- (4.75,-3.75) -- (4.75,-4.75) -- (3.75,-4.75) -- (3.75,-3.75);
      \node at (4.25,-4.25) {$+$};
      
      \node at (3,-3) {$\ddots$};
      \draw[dotted] (1,-0.5) -| (1.75,-1.25);
      \draw[dotted] (0.5,-1) |- (1.25,-1.75);
      \draw[dotted] (2.25,-1.75) -| (3,-2.5);
      \draw[dotted] (2.5,-3) -| (1.75,-2.25);
      \draw[dotted] (4.25,-3.75) |- (3.5,-3);
      \draw[dotted] (3,-3.5) |- (3.75,-4.25);
      \draw[dotted] (0,-0.5) -- (-0.2,-0.5);
      \draw[dotted] (0.5,0.2) -- (0.5,0);
      \draw[dotted] (4.75,-4.25) -- (4.95,-4.25);
      \draw[dotted] (4.25,-4.75) -- (4.25,-4.95);
      \node[smallblack,label=above:{$v_0$}] at (0.5,0.2){};
      \node[smallblack,label=above:{$v_1$}] at(1.75,-0.5){};
      \node[smallblack,label=above:{$v_2$}] at(3,-1.75){};
      \node[smallblack,label=above:{$v_{d-2}$}]at (4.25,-3){};
      \node[smallblack,label=right:{$v_{d-1}$}] at(4.95,-4.25){};
      \node[smallblack,label=below:{$v_d$}] at(4.25,-4.95){};
      \node[smallblack,label=below:{$v_{d+1}$}] at(3,-4.25){};
      \node[smallblack,label=below:{$v_{2d-3}$}] at(1.75,-3){};
      \node[smallblack,label=below:{$v_{2d-2}$}] at(0.5,-1.75){};
      \node[smallblack,label=left:{$v_{2d-1}$}] at(-0.2,-0.5){};

      \node[draw,circle, minimum size = 2cm] (m) at (7,-2) {$G_v$};
      \draw[white,bend left=70] ($(m)+(-0.81,-0.58)$) to node[sloped,midway] {\textcolor{black}{$\dots$}} ($(m)+(0.31,0.95)$);
      \node[smallblack,label=left:{$v_0$}] at ($(m)+(1,0)$){};
      \node[smallblack,label=above:{$v_1$}] at ($(m)+(0.31,-0.95)$){};
      \node[smallblack,label=right:{$v_2$}] at ($(m)+(-0.81,-0.58)$){};
      \node[smallblack,label=below:{$v_{2d-1}$}] at ($(m)+(0.31,0.95)$){}; 
    \end{tikzpicture}
    \caption{Vertex gadget $G_v$ (dotted links mean identification of vertices) and its schematic representation.}
    \label{fig:wholev}
  \end{figure}

  We use again the crossing gadget to define the edge gadget. For each
  edge $vw$ of $G$, we link two pairs $(v_{2i},v_{2i+1})$ and
  $(w_{2j},w_{2j+1})$ of consecutive ports of the vertex gadgets $G_v$ and $G_w$
  by identifying $v_{2i},v_{2i+1},w_{2j}$ with $x_1,y_1,x_2$, and adding
  an alternating path of length~$2$ between $y_2$ and $w_{2j+1}$ as
  shown in Figure~\ref{fig:wholee}. 

  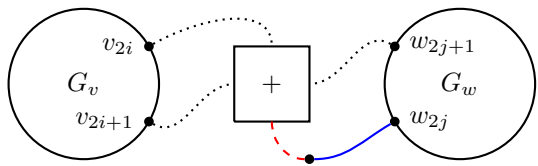
\begin{figure}[!ht]
    \centering
    \begin{tikzpicture}[thick]
      \node[draw,circle, minimum size = 2cm] (m) at (0,0) {$G_v$};
      \node[smallblack,label=left:{$v_{2i}$}] (v2) at (0.866,0.5){};
      \node[smallblack,label=left:{$v_{2i+1}$}] (v1) at (0.866,-0.5){};
      \draw[out=30,in=90,dotted] (v2) to (2.5,0.5);
      \draw[out=-30,in=180,dotted] (v1) to (2,0);

      \tikzset{xshift=2cm}
      \draw (0,-0.5) -- (1,-.5) -- (1,0.5) -- (0,0.5) -- (0,-0.5);
      \node at (0.5,0) {$+$};

      \tikzset{xshift=3cm}

      \node[draw,circle, minimum size = 2cm] (m) at (0,0) {$G_w$};
      \node[smallblack,label=right:{$w_{2j}$}] (w2) at (-0.866,-0.5){};
      \node[smallblack,label=right:{$w_{2j+1}$}] (w1) at (-0.866,0.5){};
      \node[smallblack] (x) at (-2,-1) {};
      \draw[in=0,out=210,blue] (w2) to (x);
      \draw[out=180,in=270,red,dashed] (x) to (-2.5,-0.5);
      \draw[out=150,in=0,dotted] (w1) to (-2,0);

    \end{tikzpicture}
    \caption{Edge gadget (dotted links mean identification of
      vertices).}
    \label{fig:wholee}
  \end{figure}

  Due to this construction, any homomorphism $H\to \rho(UC_4)$ must
  map the bottom vertex ($x_2$) from the crossing gadget and $w_{2j+1}$ to different images, with a difference in
  $\{\pm 2, \pm 4\}$. Thus, Lemma~\ref{lem:cross} also ensures that
  Condition~\ref{cond:proper} holds since any homomorphism maps
  $v_{2i+1}$ and $w_{2j+1}$ to the same image (the ground) and $v_{2i}$ and
  $w_{2j}$ to different ones (their colours).

  To ensure that this construction outputs a planar graph $H$, we
  first fix a planar embedding of $G$. Then, we use the ports of the
  gadgets associated to $v$ in the same cyclic ordering as the edges
  incident to $v$ in $G$.

\subsection{End of the proof}

As we already saw, our three conditions are satisfied by $H$, hence if there
is a homomorphism $H\to \rho(UC_4)$, then $G$ is $3$-colourable.
Conversely, given a $3$-colouring of $G$, we can define a homomorphism
$\varphi:H\to \rho(UC_4)$. Given a vertex $v$ of $G$ with colour
$c\in\{1,2,3\}$, we define $\varphi(v_{2i})=u_0$ and
$\varphi(v_{2i+1})=u_{2c}$. We can then extend this partial
homomorphism to each crossing gadget. Therefore, we obtain a
reduction. Since $H$ can be constructed in polynomial time, we finally
obtain that \PHOM{\rho(UC_4)} (and thus, \Pshom{UC_4}) is NP-hard.

\section{Planar graphs with large girth}\label{sec:girth}

In this section, we prove a ``hypothetical complexity'' type theorem
for s-homomorphisms to unbalanced even cycles, similar to the one
in~\cite{EMOP12} for \PHOM{C_{2k+1}}. This is motivated by the signed graph analogue of Jaeger's conjecture: it is conjectured in~\cite{NRS14} that every planar bipartite signed graph
of girth at least~$4k-2$ has an s-homomorphism to $UC_{2k}$ (see~\cite{CNS} for recent progress).

\begin{theorem}\label{thm:girth}
  Let $g\ge4$ and $k\ge2$ be fixed integers.  Either every bipartite
  planar signed graph with girth at least $g$ has an s-homomorphism to
  $UC_{2k}$, or \Pshom{UC_{2k}} is NP-complete for planar signed
  graphs with girth at least $g$.
\end{theorem}
\begin{proof}
  Note that by Proposition~\ref{prop:permequiv} and
  Corollary~\ref{cor:permequiv}, the statement is equivalent to the
  one stating that every bipartite planar signed graph with girth at
  least $g$ has a homomorphism to $\rho(UC_{2k})$, or
  \PHOM{\rho(UC_{2k})} is NP-complete for planar signed graphs with
  girth at least $g$.

  Suppose that $H$ is a bipartite planar graph with girth at least $g$
  that does not map to $\rho(UC_{2k})$ and that $H$ is minimal with
  respect to the subgraph order.  Let $xy$ be a positive edge of
  $H$. By minimality, $H\setminus xy$ is $\rho(UC_{2k})$-colourable.
  Let $i$ be the smallest integer such that there exists a
  homomorphism that maps $x$ to $u_0$ and $y$ to $u_i$. So $i$ is odd
  and $3\le i\le 2k-1$.  Let $J$ be the graph obtained from
  $H\setminus xy$ by adding a path $x=x_0,x_1,\ldots,x_{i-1},x_i=y$ of
  $i$ positive edges between $x$ and $y$.  So $J$ is
  $\rho(UC_{2k})$-colourable and since $\rho(UC_{2k})$ is
  edge-transitive, we can assume that $x_0$ maps to $u_0$ and $x_1$
  maps to $u_1$. Then by minimality of $i$, $x_j$ maps to $u_j$ for
  every $0\le j\le i$ and we call this a canonical colouring.  Now we
  consider the graph $J'$ obtained from two copies $J_1$ and $J_2$ of
  $J$ by identifying the vertices $x_{i-1}$ (resp. $x_i$) of both
  copies. If a canonical colouring of the subgraph $J_1$ is extended
  to $J'$, then both vertices $x_0$ map to $u_0$.  Thus, every
  $\rho(UC_{2k})$-colouring of $J'$ maps both vertices $x_0$ to the
  same vertex.  Notice that the distance between the two vertices
  $x_0$ is $2i-2\ge 4$.

  Then, every instance $G$ of \PHOM{\rho(UC_{2k})} can be transformed
  into an equivalent instance $G'$ of \PHOM{\rho(UC_{2k})} with girth
  at least $g$ using, as a vertex gadget, sufficiently many copies of the gadget $J'$.
\end{proof}

\section{Restricting the maximum degree}\label{sec:maxdeg}

We now consider instance restrictions according to the maximum
degree. It is proved in~\cite{GHN00} that \HOM{C_{2k+1}} can be
reduced to \HOM{C_{2k+1}} for subcubic graphs using a gadget consisting of a sequence of copies of $C_{2k+1}$ glued to each other. Furthermore, this reduction preserves the planarity,
thus it follows from~\cite{HNT06,MS09} that \PHOM{C_{2k+1}} is
NP-complete for subcubic graphs. Here, we show an analogue of this
result for s-homomorphisms to unbalanced even cycles.

\begin{theorem}\label{thm:girth-smalldeg}
  \Pshom{UC_4} remains NP-complete for signed graphs with maximum
  degree~$4$. For every $k\geq 3$, \Pshom{UC_{2k}} remains NP-complete for
  subcubic signed graphs (of girth~$2k$).
\end{theorem}
\begin{proof}
  In both cases, we reduce \Pshom{UC_{2k}} itself to \Pshom{UC_{2k}} on graphs with maximum
  degree $\Delta=3$ or $4$ using a vertex-gadget with appropriate vertex degrees that
  forces the same colour (say colour $0$) on arbitrarily many vertices
  of degree $\Delta-1$. Our gadgets are similar to the ones from~\cite{GHN00} for
  \PHOM{C_{2k+1}}, and are depicted in Figure~\ref{fig:huc}. For the reduction, each vertex $v$ of degree~$d$ of the input graph is replaced by a copy $G_v$ of this gadget containing $2d$ glued cycles, and if $v$ and $w$ are adjacent, we add an edge between a vertex of $G_v$ and a vertex of $G_w$ that are labeled $0$ in Figure~\ref{fig:huc}. Since the gadgets have girth~$2k$ and by Theorem~\ref{thm:girth}, \Pshom{UC_{2k}} is NP-complete for inputs of girth~$2k$, our construction producees inputs of girth~$2k$.

Observe that
  this approach does not work for \Pshom{UC_4} on subcubic graphs:
  vertices coloured with $0$ in the corresponding gadget have degree~$3$,
  not~$2$.\end{proof}

  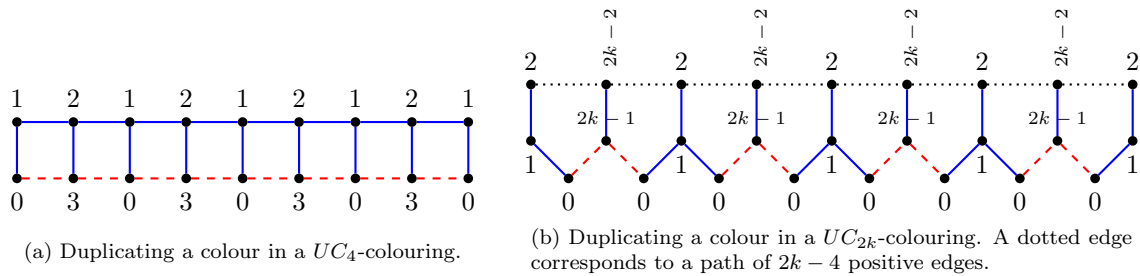
\begin{figure}[!ht] 
    \begin{tabular}{cc}
      \begin{tikzpicture}[thick,scale=0.75]
        \node[smallblack,label=above:{$1$}] (01) at (0,1){}; 
        \node[smallblack,label=above:{$2$}] (11) at (1,1){}; 
        \node[smallblack,label=above:{$1$}] (21) at (2,1){}; 
        \node[smallblack,label=above:{$2$}] (31) at (3,1){}; 
        \node[smallblack,label=above:{$1$}] (41) at (4,1){}; 
        \node[smallblack,label=above:{$2$}] (51) at (5,1){}; 
        \node[smallblack,label=above:{$1$}] (61) at (6,1){}; 
        \node[smallblack,label=above:{$2$}] (71) at (7,1){}; 
        \node[smallblack,label=above:{$1$}] (81) at (8,1){}; 
        \node[smallblack,label=below:{$0$}] (00) at (0,0){}; 
        \node[smallblack,label=below:{$3$}] (10) at (1,0){}; 
        \node[smallblack,label=below:{$0$}] (20) at (2,0){}; 
        \node[smallblack,label=below:{$3$}] (30) at (3,0){}; 
        \node[smallblack,label=below:{$0$}] (40) at (4,0){}; 
        \node[smallblack,label=below:{$3$}] (50) at (5,0){}; 
        \node[smallblack,label=below:{$0$}] (60) at (6,0){}; 
        \node[smallblack,label=below:{$3$}] (70) at (7,0){}; 
        \node[smallblack,label=below:{$0$}] (80) at (8,0){}; 
        \draw[red,dashed] (00) -- (10) -- (20) -- (30) -- (40) -- (50) -- (60) -- (70) -- (80);
        \draw[blue] (00) -- (01) -- (11) -- (21) -- (31) -- (41) -- (51) -- (61) -- (71) -- (81) -- (80);
        \draw[blue] (10) -- (11);
        \draw[blue] (20) -- (21);
        \draw[blue] (30) -- (31);
        \draw[blue] (40) -- (41);
        \draw[blue] (50) -- (51);
        \draw[blue] (60) -- (61);
        \draw[blue] (70) -- (71);
      \end{tikzpicture}
      & \begin{tikzpicture}[thick]
        \node[smallblack,label=above:{$2$}] (01) at (0,0.75){}; 
        \node[smallblack,label={[label distance=1mm,rotate=90]right:{\scriptsize $2k-2$}}] (11) at (1,0.75){}; 
        \node[smallblack,label=above:{$2$}] (21) at (2,0.75){}; 
        \node[smallblack,label={[label distance=1mm,rotate=90]right:{\scriptsize $2k-2$}}] (31) at (3,0.75){}; 
        \node[smallblack,label=above:{$2$}] (41) at (4,0.75){}; 
        \node[smallblack,label={[label distance=1mm,rotate=90]right:{\scriptsize $2k-2$}}] (51) at (5,0.75){}; 
        \node[smallblack,label=above:{$2$}] (61) at (6,0.75){}; 
        \node[smallblack,label={[label distance=1mm,rotate=90]right:{\scriptsize $2k-2$}}] (71) at (7,0.75){}; 
        \node[smallblack,label=above:{$2$}] (81) at (8,0.75){}; 
        \node[smallblack,label=below:{$1$}] (00) at (0,0){}; 
        \node[smallblack,label=above:{\scriptsize $2k-1$}] (10) at (1,0){}; 
        \node[smallblack,label=below:{$1$}] (20) at (2,0){}; 
        \node[smallblack,label=above:{\scriptsize $2k-1$}] (30) at (3,0){}; 
        \node[smallblack,label=below:{$1$}] (40) at (4,0){}; 
        \node[smallblack,label=above:{\scriptsize $2k-1$}] (50) at (5,0){}; 
        \node[smallblack,label=below:{$1$}] (60) at (6,0){}; 
        \node[smallblack,label=above:{\scriptsize $2k-1$}] (70) at (7,0){}; 
        \node[smallblack,label=below:{$1$}] (80) at (8,0){}; 
        \node[smallblack,label=below:{$0$}] (0) at (0.5,-0.5){}; 
        \node[smallblack,label=below:{$0$}] (1) at (1.5,-0.5){}; 
        \node[smallblack,label=below:{$0$}] (2) at (2.5,-0.5){}; 
        \node[smallblack,label=below:{$0$}] (3) at (3.5,-0.5){}; 
        \node[smallblack,label=below:{$0$}] (4) at (4.5,-0.5){}; 
        \node[smallblack,label=below:{$0$}] (5) at (5.5,-0.5){}; 
        \node[smallblack,label=below:{$0$}] (6) at (6.5,-0.5){}; 
        \node[smallblack,label=below:{$0$}] (7) at (7.5,-0.5){}; 
        \draw[red,dashed] (0) -- (10) -- (1);
        \draw[red,dashed] (2) -- (30) -- (3);
        \draw[red,dashed] (4) -- (50) -- (5);
        \draw[red,dashed] (6) -- (70) -- (7);
        \draw[dotted] (01) -- (11) -- (21) -- (31) -- (41) -- (51) -- (61) -- (71) -- (81);
        \draw[blue] (0) -- (00) -- (01);
        \draw[blue] (10) -- (11);
        \draw[blue] (1) -- (20) -- (21);
        \draw[blue] (2) -- (20);
        \draw[blue] (30) -- (31);
        \draw[blue] (3) -- (40) -- (41);
        \draw[blue] (4) -- (40);
        \draw[blue] (50) -- (51);
        \draw[blue] (5) -- (60) -- (61);
        \draw[blue] (6) -- (60);
        \draw[blue] (70) -- (71);
        \draw[blue] (7) -- (80) -- (81);
      \end{tikzpicture}
      \\
      \footnotesize{(a) Duplicating a colour in a $UC_4$-colouring.} 
      & \begin{minipage}{80mm}
        \footnotesize{(b) Duplicating a colour in a $UC_{2k}$-colouring. A dotted edge corresponds to a path of $2k-4$ positive edges.}
      \end{minipage}
    \end{tabular}
    \caption{Gadgets with vertices of degree~$3$ or~$2$, that are
      mapped to the same vertex in a $UC_{2k}$-colouring.} \label{fig:huc}
  \end{figure}

\section{Concluding remarks}\label{sec:conclu}

It would be interesting to settle the complexities of \shom{UC_4} and \Pshom{UC_4} for graphs of maximum degree~$3$. Studying \Pshom{H} for further classes of signed graphs $H$ is also of interest. As \PHOM{H} is connected to studies on homomorphism bounds for planar graphs~\cite{HNT06,N07,NO06}, \Pshom{H} is connected to the signed counterparts of these works, see for example~\cite{BFN,NRS13}.

When it comes to non-signed graphs, combined results of Moser~\cite{M97} and Zhu~\cite{Zhu99} show that
for every rational~$q$ such that $2\le q\le 4$, there exists a planar
graph with circular chromatic number~$q$.  The problem of deciding
whether the circular chromatic number of a planar graph is at most~$q$
is NP-complete if $q$ is equal to:
\begin{itemize}
\item $3+\frac12$, by Theorem~\ref{thm:c2t} since
  $C_7^2=\overline{C_7}=K_{7/2}$.
\item $2+\frac1t$ for every $t\ge1$ by~\cite{HNT06,MS09}
\item $2+\frac{2}{2t-1}$ for every $t\ge2$, by Theorem~\ref{thm:cc}.
\end{itemize}
It would be nice to extend these results to every rational $q$ with
$2<q<4$.

\end{document}